\documentclass[12pt]{article} 
\def\withcolors{0}
\def\withnotes{0}

\usepackage{ccanonne-arxiv}
\newcommand{\sprs}{s}
\newcommand{\nlm}{m}
\newcommand{\ui}{{t}} %
\newcommand{\ci}{{i}} %
\newcommand{\idx}[1]{({#1})}
\newcommand{\blk}{M}
\newcommand{\tri}{\{-1, 0, +1\}}

\usepackage{times}

\title{The Role of Interactivity in Structured Estimation}
\author{
Jayadev Acharya \\
Cornell University\\
\email{acharya@cornell.edu}
\and
Cl\'ement L. Canonne\\
University of Sydney\\
\email{clement.canonne@sydney.edu.au}
\and
Ziteng Sun\\
Google Research\\
\email{zitengsun@google.com}
\and
Himanshu Tyagi\\ 
Indian Institute of Science\\
\email{htyagi@iisc.ac.in}
}

\begin{document}
\maketitle
\begin{abstract}%
We study high-dimensional sparse estimation under three natural constraints: communication constraints, local privacy constraints, and linear measurements (compressive sensing). Without sparsity assumptions, it has been established that interactivity cannot improve the minimax rates of estimation under these information constraints. The question of whether interactivity helps with natural inference tasks has been a topic of active research. We settle this question in the affirmative for the prototypical problems of high-dimensional sparse mean estimation and compressive sensing, by demonstrating a gap between interactive and noninteractive protocols. 
We further establish that the gap increases when we have more structured sparsity:
for \emph{block sparsity} this gap can be as large as \emph{polynomial} in the dimensionality.
Thus, the more structured the sparsity is, the greater is the advantage of interaction.
Proving the lower bounds requires a careful breaking of a sum of correlated random variables
into independent components using Baranyai's theorem on decomposition of hypergraphs, which might be of independent interest.
\end{abstract}

\section{Introduction}
Estimating high-dimensional parameters is a central task arising in various scientific disciplines and data-driven applications. 
Modern applications often involve data from distributed or online sources which restrict the mechanism via which we have access to the data; for instance, 
limitations may be placed due to
 ease of implementation, or due to stringent communication constraints (bandwidth), or legal constraints (privacy). 
 
Understanding the interplay between these restrictions and the task at hand is the key to designing better and more efficient algorithms for these tasks. In this paper, we make progress on that front, considering  the fundamental question of \emph{sparse parameter estimation}. 
\begin{description}
    \item[Sparse mean estimation:] Upon observing independent samples $X_1,\dots,X_\ns$ from a high-dimensional product distribution over $\bool^\dims$ with mean vector $\mu\in[-1,1]^\dims$, the goal is to output an estimate $\hat{\mu}$ such that
    \begin{equation}\label{task:sparsemeanestimation}
        \bPr{\normtwo{\hat{\mu}-\mu} > \dst } \leq 1/10,
    \end{equation}
    \ie to achieve good accuracy of estimation under $\lp[2]$ loss. In addition, we are promised that the unknown parameter $\mu$ is $\sprs$-sparse; namely, it has at most $\sprs$ nonzero coordinates and $\norm{\mu}_0 \leq \sprs$.

    However, the observations are subject to an $\numbits$-bit \emph{communication constraint}, where $1\leq \numbits\leq \dims$. Namely, each $X_\ui$ must be compressed to an $\numbits$-bit message $Y_\ui$, and the estimate $\hat{\mu}$ is then computed from the $\ns$ messages $Y_1,\dots,Y_\ns\in\{0,1\}^\numbits$ only. 
    Our results also apply to \emph{local differential privacy} (LDP) constraints~\citep{DMNS:06,KLNRS:11}, where each message is required not to reveal too much about the observation; we relegate the details to the appendix.
    
    \item[Block-sparse mean estimation:] The task is very similar, but the sparsity structure imposed on $\mu$ is now more restrictive. Specifically, we are promised that the (at most) $\sprs$ nonzero coordinates of the unknown parameter $\mu$ are contiguous:
    \begin{equation}\label{task:blocksparsemeanestimation}
        \exists 1\leq \ci\leq \dims-\sprs: \forall j \notin \{\ci,\ci+1,\dots, \ci+\sprs\}, \,\, \mu\idx{j} = 0. 
    \end{equation}
    \item[Compressive sensing:]
    There is an unknown $\sprs$-sparse vector $x\in\R^\dims$, which can only be observed through noisy linear measurements given by
    \begin{equation}\label{task:compressivesensing}
        Y_\ui \eqdef A_{\ui} x + Z_\ui, \qquad 1\leq \ui \leq \ns,
    \end{equation}
    where $Z_1,\dots, Z_\ns$ are \iid $\gaussian{0}{I_\nlm}$ random
    variables (noise), and $A_{1},\dots, A_{\ns}\in \R^{\nlm\times \dims}$ are measurement unitary matrices\footnote{More generally, this can be relaxed to requiring only that each row vector has bounded $\lp[2]$ norm.} chosen (possibly adaptively) by the protocol. %
    The goal is to estimate $x$ to $\lp[2]$ loss $\dst$ using observations $Y_1, \dots, Y_{\ns}\in \R^\nlm$, minimizing the number $\nlm \cdot \ns$ of overall measurements. 
    When the matrices $A_{\ui}$ are chosen interactively, this is known as \emph{adaptive sensing}~\citep{AriasCD12}; specifically, adaptive sensing considers the case $\nlm=1$ and allows each measurement to be of the form $\dotprod{\mathbf{a}_\ui}{x} + z_\ui$ for a vector $\mathbf{a}_\ui$ that is adaptively chosen dependent on $Y_1,\ldots, Y_{\ui-1}$.
\end{description}

All these tasks have received significant attention in recent years. But the role of \emph{interactivity} in communication protocols is not completely understood. Interactivity allows clients to choose the messaging scheme based on clients' outputs from previous communication rounds. Formally, for (sequentially) \emph{interactive} protocols, the messaging scheme from $X_t$ to $Y_t$ is allowed to be chosen based on previous messages $Y_1, \ldots, Y_{t-1}$ while for \emph{noninteractive} protocols, the mapping from $X_\ui$ to $Y_\ui$ is chosen independently without observing others' messages.  Although interactivity brings flexibility in the protocol design, it often comes with extra cost. For example, interaction may lead to time delays since each client needs messages from previous clients, which can be prohibitive for large-scale distributed learning systems such as those used for Federated Learning~\citep{Kairouz21federated}.
 Despite these overheads, it is not fully understood whether \emph{interactivity} can lead to significant savings. 

We make progress in this direction and show that for the three examples above interactivity \emph{does} enable more data-efficient solutions. At a high level, our results can be interpreted as follows:
\begin{framed}
\noindent Interactivity allows one to leverage the \emph{structure} (sparsity) of the three tasks considered to obtain provably more data-efficient estimation algorithms (in a minimax sense).
\end{framed}
This is to be put in contrast to two related tasks. First, it has recently been shown that for {unstructured} estimation tasks, allowing for interactivity does not yield any speedup over noninteractive protocols, or, indeed, even over \emph{private-coin} protocols (where the users do not have access to any common random seed, but instead are fully independent)~\citep{BGMNW:16, HOW:18,AcharyaCST21}. 
That is, for unstructured estimation, neither public randomness nor interactivity are useful. Second, for communication (or local privacy) constraints such as the ones considered in this work, a sequence of papers~\citep{AcharyaCT20a,AcharyaCT20b,IIUIC,ACHST:20} showed that \emph{goodness-of-fit testing} (not estimation) could be more data-efficient when allowing for some coordination between users; however, this gain in efficiency was enabled by the use of a common random seed (i.e., public-coin protocols), a weaker setting than interactivity. Moreover, once this common random seed was available, letting the users interact would not lead to any additional saving: put differently, under those constraints public randomness helps for testing, but interactivity does not. 

We state the formal statements of our results in~\cref{ssec:results}, and then put them in context and discuss prior work in~\cref{ssec:prior}. We discuss details about sparse estimation and block-sparse estimation in Section~\ref{sec:sparse_mean} and Section~\ref{sec-block-sparse} respectively. We present extensions to adaptive sensing and estimation under local privacy constraints in the appendix. 

\subsection{Our results and contributions}
\label{ssec:results}
Our first result concerns the lower bound of noninteractive protocols for sparse mean estimation, which provides a lower that establishes a strict separation 
between the performance of interactive and noninteractive protocols:
\begin{theorem}
  \label{theo:bernoullimean:sparse}
  For any $\sprs \geq 4\log\dims$, any $\numbits$-bit noninteractive protocol for mean estimation of $\sprs$-sparse product distributions over $\bool^\dims$ must have sample complexity 
  $
      \bigOmega{\frac{\sprs\dims}{\dst^2\numbits}\log\frac{e\dims}{\sprs}}
  $.
\end{theorem}
Combined with previously known results for
sparse mean estimation (detailed in Section~\ref{sec:sparse_mean}), this lower bound immediately implies the following:
\begin{corollary}
  \label{coro:bernoullimean:sparse}
  For any $\sprs \geq 4\log\dims$, the \emph{noninteractive} sample complexity of mean estimation of $\sprs$-sparse product distributions over $\bool^\dims$ under $\numbits$-bit communication constraints is
  \smash{$
      \bigTheta{\frac{\sprs\dims}{\dst^2\numbits}\log\frac{e\dims}{\sprs}}
  $}, while the \emph{interactive} sample complexity is
  \smash{$
      \bigTheta{\frac{\sprs\dims}{\dst^2\numbits} + \frac{\sprs}{\dst^2}\log\frac{e\dims}{\sprs}}
  $}.
\end{corollary}
This shows that interactive protocols outperform the noninteractive ones by a factor up to $\Omega(\log\dims/\sprs)$. We emphasize that prior to our work this gap was only known for $\dst \ll \sqrt{{\numbits}/{\dims}}$, from~\citep{HOW:18}, \new{even for the case $\sprs=1$.}\medskip

\noindent Our second set of results focuses on  \emph{block sparsity}.
\begin{theorem}
  \label{theo:bernoullimean:blocksparse}
  For any $\sprs \geq 1$, the \emph{noninteractive} sample complexity of mean estimation of $\sprs$-block sparse product distributions over $\bool^\dims$ under $\numbits$-bit communication constraints is 
  \smash{$
      \bigTheta{\frac{\sprs\dims + \dims\log\dims}{\dst^2\numbits}}
  $}, while the \emph{interactive} sample complexity is
  \smash{$
      \tildeO{\frac{\sprs^2+\dims}{\dst^2\numbits} + \frac{\sprs}{\dst^2}}
  $} and \smash{$
      \bigOmega{\frac{\sprs^2+\dims}{\dst^2\numbits} + \frac{\sprs}{\dst^2}}
  $}.
\end{theorem}
Only a restricted version of the upper bound in the interactive case was previously known;\footnote{That is, the existing upper bound worked under an additional promise on the block sparsity, which was that all biased coordinates had the same magnitude~\citep{AcharyaCMT21}.} our results, by complementing them with the required lower bounds (as well as the noninteractive upper bound) establish that interactivity leads to significant savings under this more structured sparsity constraint. As an example, for $\sprs \approx \sqrt{\dims}$, the sample complexity for interactive protocols is $\tilde{\Theta}(\dims/(\dst^2\numbits))$ whereas that of noninteractive protocols is $\tilde{\Theta}(\dims^{3/2}/(\dst^2\numbits))$. \new{Interestingly, establishing the lower bound in the noninteractive case (Lemma~\ref{lemma:blocksparse:noninteractive}) requires circumventing many technical hurdles, and in particular handling high-order correlations between random variables when trying to bound the expectation of a multivariate polynomial with the measure change bound of Lemma~\ref{l:basic_mc}. To achieve this, we carefully decompose the dependency (hyper)graph of the resulting monomials into sums of independent terms, taking recourse to a result of~\citet{baranyai1974factorization} on factorization of hypergraphs (Lemma~\ref{lem:baranyai}). We believe this strategy to be of independent interest, with applications to other statistical lower bounds in distributed settings.}\medskip

\noindent Finally, our third set of results departs from communication constraints, and instead focuses on the well-studied question of \emph{compressive sensing}. Specifically, as discussed earlier, we consider the problem of estimating (under the $\lp[2]$ loss) an $\sprs$-sparse signal, when the only measurements allowed are $\nlm$-dimensional noisy linear measurements (as defined in~\cref{task:compressivesensing}).
\begin{theorem}
  \label{theo:compressive}
  For any $\sprs \geq 4\log\dims$, there exists an interactive protocol for compressive sensing for $\sprs$-sparse vectors using $\nlm$-dimensional noisy linear measurements with sample complexity 
  $
      \bigO{\frac{\sprs\dims}{\dst^2\nlm} + \frac{\sprs}{\dst^2}\log \frac{e\dims}{\sprs}}
  $.
\end{theorem}
Combined with known results on compressive sensing (detailed in Appendix~\ref{sec:adaptive_sensing}), our upper bound readily implies the following:
\begin{corollary}
  For any $\sprs \geq 4\log\dims$, the \emph{noninteractive} sample complexity of compressive sensing for $\sprs$-sparse vectors using $\nlm$-dimensional random measurements is
  \smash{$
      \bigTheta{\frac{\sprs\dims}{\dst^2\nlm}\log\frac{e\dims}{\sprs}}
  $}, while the \emph{interactive} sample complexity is
  \smash{$
      \bigTheta{\frac{\sprs\dims}{\dst^2\nlm}+ \frac{\sprs}{\dst^2}\log \frac{e\dims}{\sprs}}
  $}.
\end{corollary}
Taken together, our three sets of results show that, across various tasks, interactivity \emph{does} help for estimation under constraints, as soon as sparsity enters the picture. Further, it is not too hard to show that the analogues of Corollary~\ref{coro:bernoullimean:sparse}  and \cref{theo:bernoullimean:blocksparse} hold for \emph{local privacy} constraints as well, replacing $\numbits$ by square of the privacy parameter, which demonstrates corresponding separations under LDP.
\begin{theorem}[Local privacy (LDP)]
  All the bounds from Corollary~\ref{coro:bernoullimean:sparse} and \cref{theo:bernoullimean:blocksparse} hold when considering $\priv$-LDP constraints instead of $\numbits$-bit communication constraints, replacing $\numbits$ by $\priv^2$ in the corresponding expressions for any value of the privacy parameter $\priv \in (0,1]$.
\end{theorem}
We provide the necessary definitions and the proof of this theorem in~\cref{sec:ldp}.

\subsection{Prior and related work}
    \label{ssec:prior}
In the recent years, there has been a significant work on both distribution mean estimation and signal estimation under various constraints. We highlight below the most relevant to our work.

Distributed mean estimation under both communication and privacy constraints has been extensively considered~\citep{Shamir:14, ErlingssonPK14, DJW:17, HOW:18:v1, BCO:20, BGMNW:16, YeB17, AcharyaCT20a, ASZ:18:HR}. Most of these results pertain to noninteractive protocols, namely schemes where the measurements/messaging schemes are decided simultaneously, not allowing for dependence on the outcomes from prior symbols. There are some notable exceptions.~\cite{BGMNW:16,DR:19} establish interactive lower bounds for estimating high dimensional distributions under communication and local privacy constraints. Their strong results establish that the minimax rates of interactive and noninteractive schemes are the same. However, these minimax lower bounds are tight only for \emph{dense} distributions.~\cite{BGMNW:16} considered sparse high-dimensional mean estimation under communication and establish lower bounds for interactive schemes and upper bounds for noninteractive schemes; still, their result leave open the existence of a gap between the two for sparse mean estimation. Similarly,~\cite{DR:19} consider sparse mean estimation under local privacy: their work also leaves unanswered the existence of a gap between the interactive lower bounds and their noninteractive upper bounds.~\cite{Shamir:14} consider 1-sparse mean estimation for $d$-dimensional product distributions, and their bounds also have a similar gap. 

Block-sparse signals are common in several applications such as DNA microarrays, sensor networks and MIMO communication systems~\citep{ElhamifarV13, StojnicPH09, BarbotinHRV12, BaronDWSB09, GogineniN11, ShoukryT15, VorobyovGW04,BaraniukCDH10}. Estimating distributions with block-sparse means was considered in~\cite{AcharyaCMT21}. They study the constraint where one has access to a few coordinates of each sample and showed that for this constraint there is a separation between interactive and noninteractive protocols. This is in the context of first-order optimization, where they used a reduction to this mean estimation problem in order to show that adaptive processing of gradients can lead to faster convergence rates for distributed optimization. 

Compressive sensing has been immensely popular since the pioneering works of~\cite{CandesRT06, Donoho06}. Adaptive sensing, \ie choosing the measurements adaptively, was studied in~\cite{AriasCD12} for the case $\nlm=1$. Their results leave open a logarithmic (in the dimension) gap between upper and lower bounds on the number of measurements.

\anote{Himanshu, add some about compressive sensing here?}

The question of whether interactivity helps under local privacy constraints has been extensively studied, starting with the influential work of~\citet{KLNRS:08}, who designed a problem for which there show a separation between interactive and noninteractive schemes. \citet{DF:18} designed a class of Boolean functions for which learning under interactive LDP protocols is exponentially more expensive than noninteractive schemes.~\cite{DF:20} showed that exponentially more samples are needed to learn linear models with convex loss without interaction than that with, under both privacy and communication constraints.~\cite{JMNR:19} went a step further and showed that allowing for fully interactive schemes can provide an advantage over sequentially interactive schemes. \cite{ullman2018private} proves a lower bound for locally private hypothesis selection for noninteractive protocols, which can be viewed as a 1-sparse mean estimation problem.

Another line of work~\citep{pmlr-v65-agarwal17c, NEURIPS2019_0d441de7, thananjeyan2021pac} shows that interactivity brings advantage for the task of best arm identification in multi-armed bandits. The feedback model in multi-armed bandit can be simulated by a 1-bit communication protocol, hence our result would imply the same separation.
\subsection{Notation and Preliminaries.}
We use $\log$ and $\ln$ for logarithm in base 2 and natural logarithm respectively. 
Throughout the paper, we use standard asymptotic notation $\bigO{\cdot}$, $\bigOmega{\cdot}$, and $\bigTheta{\cdot}$, with asymptotics to be taken as $\dims,\sprs \gg 1$ and small $\dst$. Our lower bounds will routinely involve both Kullback--Leibler (KL) and chi-squared ($\chi^2$) divergences between probability distributions, defined as
\[
    \kldiv{\p}{\q} \eqdef \sum_{x\in\cX} \p(x)\ln\frac{\p(x)}{\q(x)}, \qquad 
    \chisquare{\p}{\q} \eqdef \sum_{x\in\cX} \frac{(\p(x)-\q(x))^2}{\q(x)}
\]
for any two distributions $\p,\q$ over a (discrete) domain $\cX$, with the convention that $0\ln0=0$. These divergences satisfy $\kldiv{\p}{\q} \leq \chisquare{\p}{\q}$. We will also require the notion of (Shannon) entropy $H(X)=-\sum_{x\in\cX}\p_X(x)\log\p_X(x)$ of a random variable $X$ with distribution $\p_X$, as well as that of the mutual information $\mutualinfo{X}{Y}$ between two random variables $X,Y$, defined as
\[
    \mutualinfo{X}{Y} \eqdef \kldiv{\p_{XY}}{\p_X\otimes\p_Y},
\]
where $\p_{XY}$, $\p_X,\p_Y$ are the joint distribution of $(X,Y)$ and the marginal distributions of $X$ and $Y$, respectively, and $\p\otimes\q$ is the product distribution with marginals $\p,\q$. We will also use the conditional mutual information $\condmutualinfo{X}{Y}{Z}$, defined as $\condmutualinfo{X}{Y}{Z} \eqdef \bE{Z}{\kldiv{\p_{XY\mid Z}}{\p_{X\mid Z}\otimes\p_{Y\mid Z}}}$ (where $\p_{XY\mid Z}$, $\p_{X\mid Z},\p_{Y\mid Z}$ are now the analogous distributions, conditioned on $Z$). For more on these notions and their properties, we refer the reader to the textbook by~\citet{CoverThomas:06}.

Throughout the paper, we often use the term \emph{channel} to refer to the probabilistic mapping from the user's observation to messages. Formally, the $t$th user selects a channel $W_t\colon \cX \to \cY$, where, for all input $x \in \cX$ and possible output $y \in \cY$,
\[
    W_t(y \mid x) = \bPr{Y_t = y \mid X_t = x}.
\]
For instance, by restricting the output space $\cY$ to satisfy $\abs{\cY} \le 2^\numbits$, the formulation captures $\numbits$-bit communication constraints. In \emph{noninteractive} protocols, users must select their channels independently without observing each other's message. In contrast, for (sequentially) \emph{interactive}\footnote{The lower bounds in the paper also holds for fully interactive protocols (the so-called \emph{blackboard model}) while the provided upper bounds only require sequential interactivity. We focus on sequentially interactive protocols in this paper for clarity of presentation.} protocols, the $t$th user can select their channel based on previous users' messages $Y_1, Y_2, \ldots, Y_{t-1}$. For both interactive and non-interactive protocols considered in this paper, we assume all users and the server have access to a public random seed $U$, which is independent of the samples.\footnote{For a formal definition, see~\citet{IIUIC}.}

We will rely in our proofs on the following measure change bound from~\cite{AcharyaCT20c}.
\begin{lemma}[A measure change bound]\label{l:basic_mc}
	Consider a random variable $X$ taking values in $\cX$. Let
	$\Phi\colon\cX\to\R^\dims$ be such that the random vector
	$\Phi(X)$ has independent coordinates and is 
	$\sigma^2$-subgaussian.
	Then, for any
	function $a\colon\cX\to[0,\infty)$ such that $\bEE{a(X)}<\infty$, we 
	have
	\[
	\normtwo{\bE{}{\Phi(X)a(X)}}^2\leq 2(\ln 2)\sigma^2\bE{}{a(X)}
	\bE{}{a(X)\ln \frac{a(X)}{\bE{}{a(X)}}}.
	\]
\end{lemma}

\section{Sparse mean estimation under communication constraints} \label{sec:sparse_mean}
We first establish~\cref{theo:bernoullimean:sparse}, thus establishing the claimed gap between interactive and noninteractive communication-constrained  sparse mean estimation.

Of the four ingredients required to prove~\cref{theo:bernoullimean:sparse} (two upper bounds, and two lower bounds), three follow from the literature; we restate them below for completeness. The following statements establish the sample complexity for interactive sparse mean estimation.

\begin{lemma}[{\cite[Proposition~2]{AcharyaCST21}}]
    \label{lemma:ub:sparse:interactive}
  For any $\sprs \geq 1$, the \emph{interactive} sample complexity of mean estimation of $\sprs$-sparse product distributions over $\bool^\dims$ under $\numbits$-bit communication constraints is
  {$
      \bigO{\frac{\sprs\dims}{\dst^2\numbits} + \frac{\sprs}{\dst^2}\log\frac{e\dims}{\sprs}}
  $}.
\end{lemma}
\begin{lemma}[\cite{BGMNW:16,AcharyaCST21}]
    \label{lemma:lb:sparse:interactive}
  For any $\sprs \geq 1$, the \emph{interactive} sample complexity of mean estimation of $\sprs$-sparse product distributions over $\bool^\dims$ under $\numbits$-bit communication constraints is
  {$
      \bigOmega{\frac{\sprs\dims}{\dst^2\numbits} + \frac{\sprs}{\dst^2}\log\frac{e\dims}{\sprs}}
  $}.
\end{lemma}
The algorithm achieving Lemma~\ref{lemma:ub:sparse:interactive} is based on successive elimination and requires interaction between clients. Turning to noninteractive estimation, similar bounds can be obtained, but with an extra logarithmic factor. 
\begin{lemma}
    \label{lemma:ub:sparse:noninteractive}
  For any $\sprs \geq 1$, the \emph{noninteractive} sample complexity of mean estimation of $\sprs$-sparse product distributions over $\bool^\dims$ under $\numbits$-bit communication constraints is
  {$
      \bigO{\frac{\sprs\dims\log(e\dims/\sprs)}{\dst^2\numbits}}
  $}.
\end{lemma}
\begin{proof}
The key observation is that, by using the ``simulate-and-infer'' idea of~\citet{AcharyaCT20b} (specifically used in the context of product distributions over $\bool^\dims$ in~\citet{AcharyaCT20c}), it suffices to show an
\smash{$
      \bigO{\frac{\sprs\log(e\dims/\sprs)}{\dst^2}}
  $} 
  upper bound in the \emph{unconstrained} setting (where all the observations are fully available), as any such algorithm can be simulated by a private-coin protocol under $\numbits$-bit communication constraints at the cost of a factor $\dims/\numbits$ in the sample complexity. The idea is to partition $\dims$ coordinates into $\clg{\dims/\numbits}$ blocks of size at most $\numbits$. Then $\clg{\dims/\numbits}$ users can send their observation within each block using $\numbits$ bits. By independence of the coordinates, we get a valid sample from the source distribution by combining the messages.
  With samples from the original distribution, the $
      \bigO{\frac{\sprs\log(e\dims/\sprs)}{\dst^2}}
  $ sample complexity upper bound, in turn, is well-known, and is attained by \eg the maximum likelihood estimator. See, for instance,~\citet[Section~20.2]{Wu20}.
\end{proof}

The final component needed to
show the additional logarithmic factor is necessary
 is the noninteractive sample complexity lower bound. As discussed earlier, the required lower bound is shown in~\citet[Theorem~3]{HOW:18}, but under the restriction that $\ns \geq \frac{\sprs\dims^2\log(e\dims/\sprs)}{\numbits^2}$, making the lower bound vacuous unless $\dst \ll \sqrt{\numbits/\dims}$. We provide a proof of this lower bound, which removes this restriction on $\ns$. The crux in removing this regularity condition is to handle the dependent terms in the obtained information bound (\cref{eq:sum:information}) directly through careful conditioning; while previous techniques consider linearization of the information vector, which results in loose bounds. 
\begin{lemma}
    \label{lemma:lb:sparse:noninteractive}
  For any $\sprs \geq 4\log\dims$, the \emph{noninteractive} sample complexity of mean estimation of $\sprs$-sparse product distributions over $\bool^\dims$ under $\numbits$-bit communication constraints is
  {$
      \bigOmega{\frac{\sprs\dims\log(e\dims/\sprs)}{\dst^2\numbits}}
  $}.
\end{lemma}
\begin{proof}
Consider the following set of $\sprs$-sparse product distributions, which we will use as the ``hard instances'' for our lower bound. Setting $\gamma \eqdef \frac{\dst}{\sqrt{\sprs}}$, for any $z\in\tri^\dims$ we define $\theta_z\in\R^\dims$ by \new{$\theta_{z,i} = \gamma z\idx{i}$}, $i\in[\dims]$. 
Let $Z$ be a random variable on $ 
\tri^\dims$ satisfying 
\[
    \bPr{Z\idx{i} = +1} = \frac{\sprs}{4\dims}, \;\;\; \bPr{Z\idx{i} = -1} = \frac{\sprs}{4\dims},\;\;\; \bPr{Z\idx{i} = 0} = 1 - \frac{\sprs}{2\dims}.
\]
Note that $\theta_Z$ is then $(\sprs/2)$-sparse in expectation, and further $\bEE{Z\idx{i}}=0$, $\sigma^2 \eqdef \bEE{Z\idx{i}^2}=\frac{\sprs}{2\dims}$ for all $i$. By a Chernoff bound, we also get that $\theta_Z$ is $\sprs$-sparse with high probability:
if $\sprs \geq 4\log\dims$, 
$
  \bPr{ \norm{\theta_Z}_0 \leq \sprs } \geq 1-\frac{\sprs}{4\dims}
$.
This will be enough for our purposes, and allows us to consider the random prior of hard instances above (product distributions over $\tri^\dims$, with mean $\theta_Z$ for random $Z$ with independent coordinates) instead of enforcing $\sprs$-sparsity with probability one (details follow).

Consider the following generative process. First pick $Z$ at 
	random from $\tri^\dims$ as above. Then, each of the  
	$\ns$ users observes one sample $X_\ui$ from the product distribution $\p_{Z}$ with mean vector $\theta_{Z}$ and sends 
	its samples through a channel $W_\ui\colon\bool^\dims\to\{0,1\}^\numbits$ to compress it to a message $Y_\ui$.
	
	The next claim states that any sufficiently accurate estimation protocol must provide enough information about each $Z\idx{i}$ 
	from the tuple of messages $Y^\ns$.
	\begin{claim}[Assouad-type Bound] \label{lem:fano}
		For any protocol that estimates $\sprs$-sparse product distributions to $\lp[2]$ accuracy 
		$\dst/4$, we must have
		$
		\sum_{i=1}^\dims \mutualinfo{Z\idx{i}}{Y^\ns} = \Omega\big(\sprs \log\frac{e\dims}{\sprs}\big).
		$
		In particular, by independence of the coordinates of $Z$, this implies 
		\smash{$
		\mutualinfo{Z}{Y^\ns} = \Omega\big(\sprs \log\frac{e\dims}{\sprs}\big).
        $}	    
	\end{claim}
	\begin{proofof}{\cref{lem:fano}}
	Fix any such protocol, and consider the corresponding estimator $\hat{\theta}=\hat{\theta}(Y^\ns)$. From there, define an estimator $\hat{Z}$ for $Z$ by choosing
	\[
	        \hat{Z} = \arg\!\min_{z\in\tri^\dims} \normtwo{\theta_z-\hat{\theta}}\,.
	\]
	In particular,
	$
	    \normtwo{\theta_{\hat{Z}}-\theta_{\vphantom{\hat{Z}}Z}} \leq 2\normtwo{\hat{\theta}-\theta_Z}
	$ with probability $1$, and
	\[
	        \bEE{\normtwo{\theta_{\hat{Z}}-\theta_{\vphantom{\hat{Z}}Z}}^2}
	        \leq \bEE{\normtwo{\theta_{\hat{Z}}-\theta_{\vphantom{\hat{Z}}Z}}^2 \indic{\norm{\theta_Z}_0 \leq \sprs}} + \newest{\frac{\sprs}{4\dims}}\cdot \max_{z,z'} \normtwo{\theta_z-\theta_{z'}}^2
	        \leq 2\cdot\frac{\dst^2}{16} + \frac{\sprs}{4\dims}\cdot \frac{\dst^2}{\sprs}\cdot \dims = \frac{3\dst^2}{8},
	\]
	where we used the fact that $\hat{\theta}$ has the guarantees of a good estimator (to $\lp[2]$ loss $\dst/4$) whenever $\theta_Z$ is $\sprs$-sparse, our bound on the probability that $Z$ is not $\sprs$-sparse, and the fact that the maximum distance between any two of the mean vectors $\theta_z,\theta_{z'}$ from our construction \newest{is $\dst/\sqrt{\dims}$}. Since $\normtwo{\theta_{\hat{Z}}-\theta_{\vphantom{\hat{Z}}Z}}^2 = \frac{\dst^2}{\sprs} \sum_{i=1}^\dims \indic{Z\idx{i} \neq \hat{Z}\idx{i}}$, this implies
	\[
	        \sum_{i=1}^\dims \bPr{Z\idx{i}\neq \hat{Z}\idx{i}} \leq \frac{3\sprs}{8}\,.
	\]
	By the data processing inequality, considering the Markov chain 
	$Z\idx{i} - Y^\ns-\hat{Z}\idx{i}$, we have
	\[
	        \sum_{i=1}^\dims \mutualinfo{Z\idx{i}}{\hat{Z}\idx{i}} \leq \sum_{i=1}^\dims \mutualinfo{Z\idx{i}}{Y^\ns}\,.
	\]
	Thus, it is enough to show that $\sum_{i=1}^\dims \mutualinfo{Z\idx{i}}{\hat{Z}\idx{i}} = \Omega\Paren{\sprs \log\frac{e\dims}{\sprs}}$. Towards that, we have by Fano's inequality  
	that for all $i$
	$
	\mutualinfo{Z\idx{i}}{\hat{Z}\idx{i}} = H(Z\idx{i}) - H(Z\idx{i}\mid \hat{Z}\idx{i}) \geq h\Paren{\frac{\sprs}{2\dims}} -h(\bPr{Z\idx{i} \neq \hat{Z}\idx{i}})
	$, where $h(x)=-x\log x - (1-x)\log(1-x)$ is the binary entropy. It follows that
	\begin{align*}
	    \sum_{i=1}^\dims \mutualinfo{Z\idx{i}}{\hat{Z}\idx{i}}
	    &\geq \dims\Paren{ h\Paren{\frac{\sprs}{2\dims}}-\frac{1}{\dims}\sum_{i=1}^\dims h\Paren{\bPr{Z\idx{i} \neq \hat{Z}\idx{i}}} } \\
	    &\geq \dims\Paren{ \Paren{\frac{\sprs}{2\dims}}-h\Paren{\frac{1}{\dims}\sum_{i=1}^\dims \bPr{Z\idx{i} \neq \hat{Z}\idx{i}}} }\\
	    &\geq \dims\Paren{  h\Paren{\frac{\sprs}{2\dims}}- h\Paren{\frac{3\sprs}{8\dims}} }
	    \geq \frac{3}{100}\sprs\log \frac{e\sprs}{\dims}\,,
	\end{align*}
	where the second inequality by concavity and  monotonicity (on $[0,1/2]$) of $h$, respectively, and the last by observing that
	\[
	        \inf_{x\in[0,1]}\frac{h(x/2)-h(3x/8)}{x\log(e/x)} > 0.03\,.
	\]
	This concludes the proof.
	\end{proofof}
	
	The next (key) claim below states that, under communication constraints, the mutual 
information scales as the total number of bits communicated from the 
users. 
\begin{claim}\label{lem:mi_bound}
	For any noninteractive protocol with $\numbits$ bits from each of the $\ns$ users, we must have
	$
		\mutualinfo{Z}{Y^\ns} = O\Paren{\frac{\ns \dst^2 \numbits }{\dims}}.
	$
\end{claim}
\begin{proof}
First, we note that while the noninteractive protocol might allow for public randomness $U$ shared between users (public-coin protocols), it is enough to establish the bound for private-coin protocols. This is because we can condition on a particular realization $u$ of the public randomness $U$: by obtaining a uniform upper bound on $\condmutualinfo{Z}{Y^\ns}{U=u}$ for all $u$, the same applies to the  conditional mutual information $\condmutualinfo{Z}{Y^\ns}{U} =  \mutualinfo{Z}{Y^\ns, U}$
which is the quantity of interest.

With that in mind, note that for private-coin protocols the messages $Y_1, Y_2, 
\ldots, Y_\ns$ are mutually independent conditioned on $Z$. This implies that
\[
	\mutualinfo{Z}{Y^\ns} \le \sum_{t = 1}^\ns 	\mutualinfo{Z}{Y_t}\,,
\]
and thus it is enough to bound each term of the sum as
$
	\mutualinfo{Z}{Y_t} =  \bigO{\dst^2 \numbits/\dims}.
$
To do so, fix any $1\leq t\leq \ns$, and denote $\uniform$ the uniform distribution over $\bool^\dims$. For the channel $W_t\colon\bool^\dims\to\{0,1\}^\numbits$ used by user $t$, let $W_t^{\p}$ be the distribution on $\cY\eqdef \{0,1\}^\numbits$ induced by an input $X$ drawn from $\p$:
\begin{equation}
    \label{eq:def:induced:distribution}
        W_t^{\p}(y) = \bE{X\sim\p}{W_t(y\mid X)}, \qquad y\in\cY\,.
\end{equation}
We can rewrite and bound the mutual information as
\[
		\mutualinfo{Z}{Y_t}  = \bE{Z}{\kldiv{W_t^{\p_{Z}}}{W_t^{\uniform}}}
\le 
\bE{Z}{\chisquare{W_t^{\p_{Z}}}{W_t^{\uniform}}}\,. 
\]
We bound the mutual information for each user $t$ and drop the subscript $t$ from $W_t$ when it is clear from context. Expanding out the chi-square divergence, we obtain the following bound on the mutual information:
\begin{equation}
    \label{eq:sum:information}
    \!\!\!\mutualinfo{Z}{Y_t} \!\leq \!\sum_{y \in 
		\cY}  \!\Big( \sigma^2 \gamma^2 \sum_{i \in [\dims]} \!\! \frac{\bE{\uniform}{W(y\mid X)X\idx{i}
			}^2}{\bE{\uniform}{W(y\mid X)}}  \!+\! \sum_{ r= 2}^{\dims} \sigma^{2r}\gamma^{2r} \!\! \sum_{\substack{B 
			\subseteq [\dims]\\ |B|  = r}}  \!\!\frac{\bE{\uniform}{W(y\mid X)\prod_{i \in B}X\idx{i}
			}^2}{\bE{\uniform}{W(y\mid X)}} \Big),
\end{equation}
where $\gamma = \dst/\sqrt{\sprs}$ and $\sigma^2=\frac{\sprs}{2\dims}$.

We defer the proof of \cref{eq:sum:information} to \cref{sec:missing-proof}, and proceed to bound the right-hand-side. 
For the first term, since $X$ is 1-subgaussian, we can invoke Lemma~\ref{l:basic_mc} to get
\begin{align}
	\sigma^2 \sum_{y \in 
		\cY}\gamma^2 \sum_{i \in [\dims]} 
	\frac{\bE{\uniform}{W(y\mid X)X\idx{i}
		}^2}{\bE{\uniform}{W(y\mid X)}}  = \frac{\sprs}{2\dims} \gamma^2 \sum_{y \in 
	\cY}
	\frac{\norm{\bE{\uniform}{X W(y\mid X)}}_2^2}{\bE{\uniform}{W(y\mid X)}} \le 
	(\ln 2) \frac{\dst^2 \numbits}{\dims}.
\end{align}
Next we handle the second-order terms, i.e.,
\begin{align*}
	 \sigma^4	\gamma^4 \sum_{y \in 
		\cY}  \sum_{i = 1}^\dims \sum_{j \neq i}  \frac{\bE{\uniform}{W(y\mid X) X\idx{i} X\idx{j}
			}^2}{\bE{\uniform}{W(y\mid X)}},
\end{align*}
For all $i \in [\dims]$, we have
\begin{align*}
	\sum_{j \neq i}  \frac{\bE{
			\uniform}{W(y\mid X) X\idx{i} X\idx{j}
		}^2}{\bE{\uniform}{W(y\mid X)}} & \le \sum_{j \neq i} 
		\frac{\frac12\bE{
		\uniform|X\idx{i} = 1}{W(y\mid X) X\idx{j}
	}^2 + \frac12\bE{
	\uniform|X\idx{i} = -1}{W(y\mid X) X\idx{j}
}^2 }{\frac12\bE{\uniform|X\idx{i} = 1}{W(y\mid X)} + \frac12\bE{\uniform|X\idx{i} = 
-1}{W(y\mid X)} } \\
& \le \sum_{j \neq i}   \Paren{\frac{\bE{
		\uniform|X\idx{i} = 1 }{W(y\mid X) X\idx{j}
	}^2}{\bE{\uniform|X\idx{i} = 1}{W(y\mid X)}}  +\frac{\bE{
	\uniform|X\idx{i} = -1 }{W(y\mid X) X\idx{j}
}^2}{\bE{\uniform|X\idx{i} = -1}{W(y\mid X)}}  },
\end{align*}
and so
\begin{align}
	 \sigma^4 \gamma^4 \sum_{y \in 
		\cY} & \sum_{i = 1}^\dims \sum_{j \neq i}  \frac{\bE{
			\uniform}{W(y\mid X) X\idx{i} X\idx{j}
		}^2}{\bE{\uniform}{W(y\mid X)}}  \notag\\
	&\leq  \sigma^4\gamma^4 \sum_{y 
		\in 
	\cY}  \sum_{i = 1}^\dims  \sum_{j \neq i} 
	\Paren{\frac{\bE{ 
			\uniform|X\idx{i} = 1 }{W(y\mid X) X\idx{j}
		}^2}{\bE{\uniform|X\idx{i} = 1}{W(y\mid X)}}  +\frac{\bE{
			\uniform|X\idx{i} = -1 }{W(y\mid X) X\idx{j}
		}^2}{\bE{\uniform|X\idx{i} = -1}{W(y\mid X)}}  } \notag\\
	 &= %
	 \sigma^4 \gamma^4 \sum_{i = 1}^\dims  
	 \sum_{y 
		\in 
		\cY}  \Paren{\frac{\sum_{j \neq 
				i}\bE{
				\uniform|X\idx{i} = 1 }{W(y\mid X) X\idx{j}
			}^2}{\bE{\uniform|X\idx{i} = 1}{W(y\mid X)}}  +\frac{\sum_{j \neq 
			i}\bE{ 
				\uniform|X\idx{i} = -1 }{W(y\mid X) X\idx{j}
			}^2}{\bE{\uniform|X\idx{i} = -1}{W(y\mid X)}}  }  \notag\\
	&= %
	\sigma^4  \gamma^4 \sum_{i = 1}^\dims  
	\sum_{y 
			\in 
			\cY}  \Paren{\frac{\norm{\bE{
					\uniform|X\idx{i} = 1 }{W(y\mid X) X_{-i}
				}}_2^2}{\bE{\uniform|X\idx{i} = 1}{W(y\mid X)}}  +\frac{\norm{\bE{
				\uniform|X\idx{i} = -1 }{W(y\mid X) X_{-i}
			}}_2^2}{\bE{\uniform|X\idx{i} = 1}{W(y\mid X)}} }    \notag\\
	&\le 2 (\ln 2)  \sigma^4 \gamma^4 \cdot%
	\Paren{ 2 
	\dims 
	\numbits} = (\ln 2)\frac{\numbits \dst^2}{\dims} \cdot  %
	\dst^2.
\end{align}
Similarly, we can bound the $j$th-order terms as
$
	\frac{\numbits \dst^2}{\dims} \cdot (\ln 2) (\dst^2)^{j - 1}.
$
And thus, summing over all terms, we get
\[
		\mutualinfo{Z}{Y_t} \leq \frac{(\ln 2)\dst^2\numbits}{\dims}\sum_{j=1}^\infty \dst^{2(j-1)} \le \frac{2 (\ln 2)\numbits \dst^2}{\dims}.
\]
Summing over $1\leq t\leq \ns$, we get the desired result.
\end{proof}
Putting together Claims~\ref{lem:fano} and~\ref{lem:mi_bound} then completes the proof of Lemma~\ref{lemma:lb:sparse:noninteractive}.
\end{proof}
\begin{remark}
As a byproduct, the proof of Lemma~\ref{lemma:lb:sparse:noninteractive} above (for the noninteractive case) has an interesting corollary: the lower bound framework of~\citet{AcharyaCST21} for the interactive case, which proceeds by bounding a quantity termed \emph{average discrepancy}, could not possibly go through in the sparse case with 
$\sum_{i=1}^\dims \mutualinfo{Z\idx{i}}{Y^\ns}$ instead of average discrepancy. Indeed, 
if the bound of~\citet{AcharyaCST21} applied to $\sum_{i=1}^\dims \mutualinfo{Z\idx{i}}{Y^\ns}$ as well,
we would get the same lower bound above for interactive protocols, 
which in turn will contradict the upper bound of Lemma~\ref{lemma:ub:sparse:interactive} for interactive protocols.
\end{remark}
Combining Lemmas~\ref{lemma:ub:sparse:noninteractive},~\ref{lemma:ub:sparse:interactive},~\ref{lemma:lb:sparse:interactive}, and~\ref{lemma:lb:sparse:noninteractive} establishes~\cref{theo:bernoullimean:sparse}. 
Finally, we mention that while our noninteractive lower bound (Lemma~\ref{lemma:lb:sparse:noninteractive}) requires $\sprs=\Omega(\log\dims)$, we are able to establish separately the case $\sprs=1$ via a simple, different proof (see~\cref{theo:bernoullimean:1sparse}). We provide this result in~\cref{app:bernoullimean:1sparse}, as we believe it to be of independent interest and will also be requiring it in the proof of Lemma~\ref{lemma:blocksparse:noninteractive}.\bigskip

\section{Block-sparse mean estimation under communication constraints}\label{sec-block-sparse}
In this section, we establish~\cref{theo:bernoullimean:blocksparse}, our result for $\sprs$-block-sparse mean estimation under $\numbits$-bit communication constraints. In order to establish the result, we need an upper and a lower bound on the sample complexity of both noninteractive and interactive protocols.

Of these four bounds, only a restricted version of the interactive upper bound was known, which assumed that all coordinates of the block-sparse mean had the same magnitude and that $\numbits=1$~\citep[Theorem~13]{AcharyaCMT21}. While the algorithm can easily be made to extend to $\numbits > 1$, it crucially relies on the former assumption on the structure of the block-sparse mean, and thus does not translate to our setting.\medskip

We proceed to prove separately the four bounds,  starting with the noninteractive upper bound.
\begin{lemma}
    \label{lemma:blocksparse:noninteractive:ub}
  For any $\sprs \geq 1$, the \emph{noninteractive} sample complexity of mean estimation of $\sprs$-block sparse product distributions over $\bool^\dims$ under $\numbits$-bit communication constraints is
  \smash{$
      \bigO{\frac{\sprs\dims + \dims\log\dims}{\dst^2\numbits}}
  $}.
\end{lemma}
\begin{proof}
As in the proof of~\cref{lemma:ub:sparse:noninteractive}, by using the ``simulate-and-infer'' idea of~\citet{AcharyaCT20b} it suffices to show an
\smash{$
      \bigO{\frac{\sprs + \log\dims}{\dst^2}}
  $} 
  upper bound in the \emph{unconstrained} setting (where all the observations are fully available). 
  This $
      \bigO{\frac{\sprs + \log\dims}{\dst^2}}
  $ sample complexity upper bound then can be obtained by the following simple estimator: partition the $\dims$ coordinates in $\clg{\dims/\sprs}$ consecutive blocks of (at most) $\sprs$ coordinates, and, using the same samples, separately estimate the $\clg{\dims/\sprs}$ mean subvectors to $\lp[2]$ loss $\dst^2/3$, with probability of success $\delta \eqdef \sprs/(10\dims)$. This can be done with
  \[
        \bigO{ \frac{\sprs + \log(1/\delta)}{\dst^2} } = \bigO{ \frac{\sprs + \log\frac{\dims}{\sprs}}{\dst^2} }
        = \bigO{ \frac{\sprs + \log\dims}{\dst^2} }
  \]
  samples, by (sub) Gaussian concentration of measure. By a union bound, all of the $\clg{\dims/\sprs}$ estimates are simultaneously accurate, with probability at least $9/10$. Since the ``true'' block overlaps at most 2 consecutive blocks of the $\clg{\dims/\sprs}$ considered, it then suffices to output the vector $\hat{\mu}$ consisting of only the two estimated subvectors with largest magnitude (and all other coordinates set to zero).
\end{proof}
\noindent Next, we establish a matching lower bound for noninteractive protocols. 
\begin{lemma}
    \label{lemma:blocksparse:noninteractive}
  For any $\sprs \geq 1$, the \emph{noninteractive} sample complexity of mean estimation of $\sprs$-block sparse product distributions over $\bool^\dims$ under $\numbits$-bit communication constraints is
  \smash{$
      \bigOmega{\frac{\sprs\dims + \dims\log\dims}{\dst^2\numbits}}
  $}.
\end{lemma}
\begin{proof}
The $\bigOmega{\frac{\dims\log\dims}{\dst^2\numbits}}$ term follows from the $1$-sparse estimation lower bound established in~\cref{theo:bernoullimean:1sparse}, since any $1$-sparse product distribution is $\sprs$-block-sparse for every $\sprs$. We thus focus on the main term, and establish the $\bigOmega{\frac{\sprs\dims}{\dst^2\numbits}}$ lower bound. 

To do so, consider the following set of $\sprs$-block sparse distributions. Partition $[\dims]$ into $b \eqdef 
	\dims/\sprs$ consecutive nonoverlapping blocks, $B_1, B_2, \ldots, 
	B_b$, each of size at most $\sprs$. 
	For all $z \in \bool^\dims$ and $j \in [b]$, define $\p_{z,j}$ as 
	a product distribution over $\bool^\dims$ with mean $\theta_{z,j}$ given by
	\begin{equation} \label{eqn:def_block_sparse}
		\theta_{z, j}(i) = \begin{cases}
			\frac{\dst }{\sqrt{\sprs}} z\idx{i}, & \text{ if } i \in B_j, \\
			0, & \text{ otherwise.}
		\end{cases}
	\end{equation}
	Consider the following generative process. First independently pick $Z$ uniformly at 
	random from $\bool^\dims$ and $J$ uniformly from $[b]$. Then, each of the  
	$\ns$ users observes one sample $X_\ui$ from the product distribution $\p_{Z,J}$ with mean vector $\theta_{Z,J}$ and sends 
	its samples through a channel $W_\ui\colon\bool^\dims\to\{0,1\}^\numbits$ to compress it to a message $Y_\ui$.
	
	The next result states that any sufficiently accurate estimation protocol must provide enough information about each $Z\idx{i}$ 
	from the tuple of messages $Y^\ns$, even if $J$ is known.
	\begin{lemma}[Assouad-type Bound] \label{lem:sum_cond_mi}
		For any protocol that estimates $\sprs$-block-sparse product distributions to $\lp[2]$ accuracy 
		$\dst$, we must have
		$
		\sum_{i=1}^{\dims}\condmutualinfo{Z\idx{i}}{Y^\ns}{J} = 
		\Omega\Paren{\sprs}.
		$
	\end{lemma}
	\begin{proof}
	Let $\hat{Z}$ be an estimator of $Z$ based on $Y^\ns$. By the data 
	processing inequality, it is be enough to prove
	$
	\sum_{i=1}^{\dims}\condmutualinfo{Z\idx{i}}{\hat{Z}\idx{i}}{J} = 
	\Omega\Paren{\sprs}.
	$
	By definition,
	\begin{align}
		\sum_{i=1}^{\dims}\condmutualinfo{Z\idx{i}}{\hat{Z}\idx{i}}{J} = 
		\frac{\sprs}{\dims}\sum_{j = 1}^{b} 
		\sum_{i=1}^{\dims}\condmutualinfo{Z\idx{i}}{\hat{Z}\idx{i}}{J = j} = 
		\frac{\sprs}{\dims} \sum_{j = 
		1}^{b} \sum_{i \in B_j} \condmutualinfo{Z\idx{i}}{\hat{Z}\idx{i}}{J = j}.
		\label{eqn:sum_cond_mi}
	\end{align}
	Now,
	\begin{align}
		\sum_{j = 1}^{b} \sum_{i \in B_j} \condmutualinfo{Z\idx{i}}{\hat{Z}\idx{i}}{J = 
		j} & 
		= 
		\sum_{j = 1}^{b}  \sum_{i \in B_j}\Paren{ \condentropy{Z\idx{i}}{J = j} -  
		\condentropy{Z\idx{i}}{\hat{Z}\idx{i}, J = j}} \notag\\
	& \ge \sum_{j = 1}^{b}  \sum_{i \in B_j}\Paren{ \condentropy{Z\idx{i}}{J = 
	j} -  
		h\Paren{\bPr{\hat{Z}\idx{i} \neq Z\idx{i} \mid J = j}}} \notag\\
	& = \dims - \sum_{j = 1}^{b}  \sum_{i \in B_j} h\Paren{\bPr{\hat{Z}\idx{i} 
	\neq 
	Z\idx{i} 
	\mid J = j}}\notag\\
	& \ge \dims - \dims \cdot h\Paren{\frac1\dims \sum_{j = 1}^{b} 
	\sum_{i 
	\in 
	B_j}  \bPr{\hat{Z}\idx{i} \neq Z\idx{i} 
			\mid J = j}}\,, \label{eqn:ent_err}
	\end{align}
	where $h\colon[0,1]\to\R$ is the binary entropy function. 
By construction, for any valid protocol, we must have
\begin{align*}
	\frac{\dst^2}{10} & \ge \bEE{\sum_{i = 1}^\dims \frac{\dst^2}{\sprs} 
	\indic{\hat{Z}\idx{i} 
	\neq Z\idx{i}}} = \frac{\dst^2}{\sprs} \sum_{i = 1}^\dims \bPr{\hat{Z}\idx{i}  \neq 
	Z\idx{i}} \\
	& =\frac{\dst^2}{\sprs}  \frac{\sprs}{\dims} \sum_{j = 1}^{b}  \sum_{i 
	=1 
	}^{\dims} \bPr{\hat{Z}\idx{i}  \neq 
	Z\idx{i} \mid J = j} \\
&\ge \frac{\dst^2}{\dims} \sum_{j = 1}^{b} \sum_{i \in B_j} 
	\bPr{\hat{Z}\idx{i} \neq 
		Z\idx{i} 
		\mid J = j},
\end{align*}
which implies
\[
	\frac{1}{\dims} \sum_{j = 1}^{b}  \sum_{i \in B_j} 
	\bPr{\hat{Z}\idx{i} \neq 
		Z\idx{i} 
		\mid J = j} \le \frac{1}{10}.
\]
Combining this with~\cref{eqn:ent_err,eqn:sum_cond_mi}  
completes the proof of the lemma.
\end{proof}

Using independence of $Z\idx{i}$'s, by additivity of mutual information this claim then implies that
\begin{align}
		\condmutualinfo{Z}{Y^\ns}{J} \geq \sum_{i=1}^{\dims}\condmutualinfo{Z\idx{i}}{Y^\ns}{J} 
		= \Omega\Paren{\sprs}.
\end{align}
Having obtained a lower bound on the mutual information, we now provide an upper bound for it; combining the two will yield our lower bound for sample complexity.
\begin{lemma} 
For any noninteractive protocol using $\numbits$-bit communication constraints, we must 
have
	\[
		\condmutualinfo{Z}{Y^n}{J} = O\Paren{\frac{\ns \dst^2\numbits }{\dims}}.
	\]
\end{lemma}
\begin{proof}
	Note that, since 
	$
	\condmutualinfo{Z}{Y^n}{J}  = \frac1b \sum_{j \in [b]} 
	\condmutualinfo{Z}{Y^n}{J = j} 
	$, it is enough to prove that
	$
	\sum_{j \in [b]} 
	\condmutualinfo{Z}{Y^n}{J = j} = \bigO{\frac{\ns\dst^2\numbits }{\sprs}}.
	$
	Similar to \cref{eq:sum:information}, the first step of the proof is to bound the mutual information at each time step. Let $\gamma \eqdef \dst/\sqrt{\sprs}$, at each user $t$, the following inequality holds.
	\begin{align}\label{eq:bigstuff}
		&\sum_{j\in [b]} \condmutualinfo{Z}{Y_t}{J=j} \notag\\&\quad \le \sum_{j \in [b]} 
		\sum_{y 
			\in 
			\cY}  \Big({\gamma^{2} \!\!\sum_{i \in B_j} \frac{\bE{X}{W(y\mid X)X\idx{i}
				}^2}{\bE{X}{W(y\mid X)}}  + \sum_{ r =2}^{\sprs} \gamma^{2r} \!\!\sum_{\substack{B\subseteq B_j\\ |B| = r}}  \frac{\bE{X}{W(y\mid X)\prod_{i \in B}X\idx{i}
				}^2}{\bE{X}{W(y\mid X)}} }\Big).
	\end{align}
We defer the proof of \cref{eq:bigstuff} to \cref{sec:missing-proof}, and proceed to bound the RHS. 
	For all $r \in [\sprs]$, let 
	\[
		\zeta_r \eqdef  \gamma^{2r}  \sum_{j \in [b]} 
		\sum_{y 
			\in 
			\cY} \sum_{B 
			\subseteq B_j, |B| = r}   \frac{\bE{X}{W(y\mid X)\prod_{i \in B}X\idx{i}
			}^2}{\bE{X}{W(y\mid X)}},
	\]
whereby we can rewrite the earlier bound as
	\begin{equation}
		\sum_{j\in [b]} \condmutualinfo{Z}{Y_t}{J=j}  \le \sum_{r \in [\sprs]}  
		\zeta_r. \label{eqn:suminformation_block} 
	\end{equation}
	We will bound each $\zeta_r$ separately.
	To bound $\zeta_1$, since $X$ is 1-subgaussian, from Lemma~\ref{l:basic_mc} 
	we have
	\begin{align}
		\!\!\sum_{j \in [b]} \sum_{y \in 
			\cY} \frac{\dst^2}{\sprs} \!\!\sum_{i \in B_j} 
		\frac{\bE{X}{W(y\mid X)X\idx{i}
			}^2}{\bE{X}{W(y\mid X)}} & = \frac{\dst^2}{\sprs}  \sum_{y \in 
			\cY}
		\frac{\norm{\bE{X}{X W(y\mid X)}}_2^2}{\bE{X}{W(y\mid X)}} \notag\\
		&\le 
		2 (\ln 2) \frac{\dst^2}{\sprs} \sum_{y \in 
			\cY} \bE{X}{W(y\mid X) \log 
			\frac{W(y\mid X)}{\bE{X}{W(y\mid X)}} } \label{eqn:subgaussian}\\
		& = 2 (\ln 2)  \frac{\dst^2 }{\sprs} I(\uniform; W) \le
		2 (\ln 2) \frac{\dst^2 \numbits}{\sprs} \notag,
	\end{align}
	where \cref{eqn:subgaussian} comes from Lemma~\ref{l:basic_mc}. Next we 
	bound 
	\begin{align*}
	\zeta_2&  =  \gamma^4   \sum_{y \in 
			\cY}  \sum_{j \in [b]} 
		\sum_{i < i' \in B_j} 
		\frac{\bE{X}{W(y\mid X)X\idx{i} X\idx{i'}
			}^2}{\bE{X}{W(y\mid X)}}.
	\end{align*}
	
	Note that each term in the summation above is a product of two entries 
	of $X$, which are not independent:  hence, we cannot use 
	Lemma~\ref{l:basic_mc} directly. To resolve this, we use the following lemma, which is a consequence of Baranyai’s Theorem~\citep{baranyai1974factorization}. 

	\begin{lemma}\label{lem:baranyai}
	    For all $j \in [b]$, ${B \subset B_j: |B| = r}$, the set of all size-$r$ subsets of $B_j$, can be partitioned into $m \le 2{s - 1\choose r - 1}$ sets $\cM_1, \cM_2, \ldots, \cM_m$ such that all subsets within each $\cM_i$ are disjoint.
	\end{lemma}
	\zmargin{Do we need to prove this?}
	Without loss of generality, assume that, for all $j \in [b]$, $B_j = \{(j-1)s + 1, 
	\dots, 
	js\}$. Using Lemma~\ref{lem:baranyai} with $r = 2$, we can rewrite $\zeta_2$ as
\begin{align*}
    	\zeta_2 & = \gamma^4   \sum_{y \in 
			\cY}  \sum_{j \in [b]} 
		\sum_{i < i' \in B_j} 
		\frac{\bE{X}{W(y\mid X)X\idx{i}X\idx{{i'}}
			}^2}{\bE{X}{W(y\mid X)}} \\ 
			& =  \gamma^4   \sum_{y \in 
			\cY}  \sum_{j \in [b]} 
		\sum_{k \in [\sprs - 1]}\sum_{(i, i') \in \cM_k} 
		\frac{\bE{X}{W(y\mid X)X\idx{(j-1)\sprs + i}X\idx{(j-1)\sprs + {i'}}
			}^2}{\bE{X}{W(y\mid X)}} \\
			& = \gamma^4  \sum_{k \in [\sprs - 1]} {\sum_{y \in 
			\cY}  \sum_{j \in [b]} 
		\sum_{(i, i') \in \cM_k} 
		\frac{\bE{X}{W(y\mid X)X\idx{(j-1)\sprs + i}X\idx{(j-1)\sprs + {i'}}
			}^2}{\bE{X}{W(y\mid X)}} }.
\end{align*}
	 Note that in the summation above, for each subset $\cM_k$, the pairwise products have disjoint entries and hence independent. Moreover, $X\idx{i} X\idx{i'}$ is 1-subgaussian as well since it is supported on $\bool$ with mean zero. 
	Proceeding from 
	above, applying Lemma~\ref{l:basic_mc}, the equation can 
	be bounded by
	\begin{align*}
		\zeta_2& = \gamma^4  \sum_{k \in [\sprs - 1]} {\sum_{y \in 
			\cY}  \sum_{j \in [b]} 
		\sum_{(i, i') \in \cM_k} 
		\frac{\bE{X}{W(y\mid X)X\idx{(j-1)\sprs + i}X\idx{(j-1)\sprs + {i'}}
			}^2}{\bE{X}{W(y\mid X)}} } \\
			& \le  2 (\ln 2)  \gamma^4  \sum_{k \in [\sprs - 1]} {  \sum_{y \in 
			\cY}  \bE{\uniform }{W(y\mid X) \log 
			\frac{W(y\mid X)}{\bE{\uniform
				}{W(y\mid X)}}} }\\
				& = 2 (\ln 2)  \gamma^4 \sum_{k \in [\sprs - 1]}  I(\uniform; W) \\
			& \le 2 (\ln 2)  \frac{\dst^2 \ell}{\sprs} \cdot \dst^2.
		\end{align*}	
	Similarly, using Lemma~\ref{lem:baranyai}, we can prove that for all $r \ge 3$, 
	$
	\zeta_r \le 2 (\ln 2) \frac{\dst^2\numbits }{\sprs} \cdot (\dst^2)^{r - 1}.
	$
	Combining these and~\cref{eqn:suminformation_block}, we get
	\[
	\sum_{j \in [b]} \condmutualinfo{Z}{Y_t}{J=j} \le 2 (\ln 2) 
	\frac{\dst^2\numbits 
	}{\sprs} \sum_{r \in [\sprs]}  (\dst^2)^{r - 1} \le 4 (\ln 2) \frac{\dst^2\numbits 
	}{\sprs}. 
	\]
	
	The claim follows from the observation above since, conditioned on $Z$, $Y_1, Y_2, 
	\ldots, Y_\ns$ are independent,\footnote{We can here ignore public randomness, as we can 
	bound the quantity under each fixed realization of the public coins.} we have
	\[
		\frac{1}{b}\sum_{j \in [b]}  \condmutualinfo{Z}{Y^\ns}{J = j}  \le \frac{1}{b}\sum_{t = 
		1}^\ns \sum_{j \in [b]}  \condmutualinfo{Z}{Y_t}{J = j} \leq \frac{1}{b} \cdot  4 (\ln 2) 
		\frac{\ns \dst^2\numbits 
		}{\sprs} = 4 (\ln 2) 
		\frac{\ns \dst^2\numbits 
		}{\dims}\,,
	\]
	showing the result.
\end{proof}
\noindent Combining the two claims concludes the proof, as this implies that one must have $\frac{\ns\dst^2\numbits}{\dims} = \Omega(\sprs)$.	
\end{proof} %
We now turn to the upper bound for interactive protocols. The algorithm has a two-stage procedure. In the first stage, users first detects the ``active" block with size $\Theta(\sprs)$. Then in the second stage, the users will focus on learning coordinates within the detected block, which needs less samples.
\begin{lemma}\label{lem-mean-est-upper}
  For any $\sprs \geq 1$, the \emph{interactive} sample complexity of mean estimation of $\sprs$-block sparse product distributions over $\bool^\dims$ under $\numbits$-bit communication constraints is
  \[
      \bigO{\frac{\sprs^2 + \dims\log(\dims/\sprs)\log(s/\eps)}{\dst^2\numbits} + \frac{\sprs\log(\dims/\sprs)\log(s/\eps)}{\dst^2}}.
  \]
\end{lemma}
\begin{proof}
    The algorithm works in two stages: \emph{detection} and \emph{estimation}. We start by partitioning the $\dims$ coordinates into $T \eqdef \clg{\dims/\sprs}$ consecutive blocks of (at most) $\sprs$ coordinates, $B_1,\dots,B_T$. Let $\mu_{B_j}$ be the mean vector restricted on block $B_j$. Since the actual support of the mean vector overlaps at most 2 such blocks, if $\normtwo{\mu}^2 > \dst^2$ there exists some $j\in[T]$ such that $\normtwo{\mu_{B_j}}^2 > \dst^2/2$. On the other hand, if $\normtwo{\mu}^2 \leq \dst^2$, then no such $j$ may exist, but our task in that case will be simpler.

    The algorithm proceeds in the following two stages:
    \begin{enumerate}
        \item \textbf{Detection:} Identify, with probability at least $19/20$, a block $B_j$ such that $\normtwo{\mu_{B_j}}^2 > \dst^2/2$, if there exists one, using
        $\bigO{\frac{\dims\log(\dims/\sprs)\log(s/\eps)}{\numbits\dst^2} + \frac{\sprs\log(\dims/\sprs)\log(s/\eps)}{\dst^2}}$
        samples. 
        This detection step is the most involved, and will constitute most of the proof below.
        \item \textbf{Estimation:} If no such block was identified, output the zero vector (which is a good estimate); otherwise, consider the union of the 3 blocks $B_{j-1}\cup B_j\cup B_{j+1}$, which has at most $3\sprs$ coordinates and contains the support of the unknown $\sprs$-sparse vector $\mu$. Use the noninteractive estimation algorithm (with ``$\dims = 3\sprs$'') to learn, with probability $19/20$, the corresponding mean with $\bigO{\frac{\sprs^2}{\min(\sprs,\numbits)\dst^2}} = \bigO{\frac{\sprs^2}{\dst^2\numbits}+\frac{\sprs}{\dst^2}}$ new samples.
    \end{enumerate}
    The overall algorithm has the claimed sample complexity and, by a union bound, is successful overall with probability $9/10$. Details for the first stage follow.\medskip
    
Our algorithm will use public randomness as follows. All users jointly draw a Rademacher vector $\xi = (\xi\idx{i})_{i\in[\dims]}$ uniformly at random. Let $\Delta \eqdef 5\sqrt{s\log (s/\eps)}$. Any given user
    computes the $T$ bits $\blk\idx{1},\dots, \blk\idx{T}$ based on $\xi$ as follows. For every $ j \in [T]$, upon observing $X$ (and knowing $\xi$) each user computes $M(j)$ based on $\bar{X}_j \eqdef \sum_{i \in B_j}X\idx{i}\xi\idx{i}$ using a
    stochastic rounding algorithm:
    \[
        M(j) = \begin{cases}
            +1, & \text{ with probability } \frac{\Delta + \operatorname{Clip}_{\Delta}\Paren{\bar{X}_j}}{2\Delta}, \\
            -1, & \text{ with probability } \frac{\Delta - \operatorname{Clip}_{\Delta}\Paren{\bar{X}_j}}{2\Delta}.
        \end{cases}
    \]
    where $\operatorname{Clip}_{\Delta}(x) \eqdef \max \{ \min \{x, \Delta \}, - \Delta \}$ denotes the clipping function on the interval $[-\Delta, \Delta]$.
It can be seen that $(\blk\idx{1},\dots, \blk\idx{T})$ follows a product distribution over $\bool^T$. Next we analyze the mean on each coordinate, conditioned on $\xi$.
\begin{align}
    \bEE{M(j) \mid \xi} & = 2 \cdot \bEE{\frac{\Delta + \operatorname{Clip}_{\Delta}\Paren{\bar{X}_j}}{2\Delta} \;\middle|\; \xi} - 1  = \frac{\bEE{\operatorname{Clip}_{\Delta}\Paren{\bar{X}_j} \;\middle|\; \xi} }{\Delta}. \label{eqn:mean_block}
\end{align}
Let $\bar{\mu}(B_j, \xi) \eqdef \bEE{\bar{X}_j \mid \xi} = \sum_{i \in B_j} \xi\idx{i} \mu\idx{i}$. 
Note that when a block $j$ does not intersect with the support of $\mu$, then $\bar{\mu}(B_j,\xi) = 0$. Further, since each $X(i)$ in $B_j$ is then symmetric, the clipping does not change the mean: thus, for any $j\in[T]$ such that $B_j$ does not intersect the support of $\mu$,
\[
    \bEE{M(j) \mid \xi} = 0.
\]
That is, we then have $\bPr{ \blk\idx{j} = 1 \mid \xi} = 1/2$ regardless of the realization of the shared random variable $\xi$.
    
Suppose now that $B_j$ \emph{does} intersect the support of the mean vector $\mu$, and specifically that $\normtwo{\mu_{B_j}}^2 > \dst^2/2$. We then show the following:
    \begin{claim}
        If $\normtwo{\mu_{B_j}}^2 > \dst^2/2$, then with probability at least $1/8$ over the choice of $\xi$, we have
        \[
           \abs{ \bEE{\blk\idx{j} \mid \xi} }  \geq \frac{\eps}{40\sqrt{s \log (s/\eps)} }.
        \]
    \end{claim}
    \begin{proof}
       We first show that before performing stochastic rounding, with probability at least $1/4$ over the randomness of $\xi$ it is the case that
       \begin{equation} \label{eqn:rademacher_mean}
            \abs{\bar{\mu}(B_j, \xi)} \ge \frac{\eps}{4}.
       \end{equation}
       To see this, notice that the second moment of $\bar{\mu}(B_j, \xi)$ is large:
       \[
       \bE{\xi}{\bar{\mu}(B_j, \xi)^2} =  \bE{\xi}{\Paren{\sum_{i\in B_j} \xi_i \mu_i }^2}
        = \sum_{i\in B_j} \mu_i^2 \ge \frac{\eps^2}{2}.
       \]
       We also can control the fourth moment of $\bar{\mu}(B_j, \xi)$ as follows:
       \[
        \bE{\xi}{\bar{\mu}(B_j, \xi)^4} = \bE{\xi}{\Paren{\sum_{i\in B_j} \xi_i \mu_i }^4} = \sum_{i \in B_j} \mu_i^4 + 6 \sum_{i <i' \in B_j }\mu_i^2 \mu_{i'}^2 \le 3\Paren{\sum_{i\in B_j} \mu_i^2 }^2.
       \]
       Hence, by the Paley--Zygmund inequality, we have 
       \begin{align*}
           \bPr{\abs{\bar{\mu}(B_j, \xi)} > \frac{\eps}{4}} \ge \bPr{\bar{\mu}(B_j, \xi)^2 > \frac{1}{8} \bEE{\bar{\mu}(B_j, \xi)^2}} \ge \frac{3}{4} \frac{\bEE{\bar{\mu}(B_j, \xi)^2}^2}{\bEE{\bar{\mu}(B_j, \xi)^4}} \ge \frac{1}{4},
       \end{align*}
       which proves that, as stated, \eqref{eqn:rademacher_mean} holds with probability at least $1/4$. Next we prove that the clipping does not affect this too much; namely, that with probability least $8/9$ over the randomness of $\xi$,
       \begin{equation} \label{eqn:clipping_error}
        \abs{\bEE{\operatorname{Clip}_{\Delta}\Paren{\bar{X}_j}\mid \xi} - \bar{\mu}(B_j, \xi)} \le \frac{\eps}{8}.
       \end{equation}
       Before proving the above statement, we note that by a union bound, both \eqref{eqn:rademacher_mean} and \eqref{eqn:clipping_error} simultaneously happen with probability at least $1/4-1/9>1/8$. Combining this with \eqref{eqn:mean_block} and the value of $\Delta$ then establishes the claim. Thus, to conclude it only remains to prove~\eqref{eqn:clipping_error}.
       
       Since $\bE{\xi}{\bar{\mu}(B_j, \xi)^2} = \sum_{i\in B_j} \mu_i^2 \le \sprs$, by Markov's inequality, with probability at least $8/9$,
        \[
            \abs{\bar{\mu}(B_j, \xi)} \le 3\sqrt{\sprs}. 
        \]
        Call this event $\mathcal{E}$. Conditioning on $\mathcal{E}$, we bound the probability that the sum gets clipped. By Hoeffding's inequality, we have
        \[
            \bPr{\bar{X}_j \notin [-\Delta, \Delta] \mid \xi,\mathcal{E} } \le \bPr{ \abs{\bar{X}_j  - \bar{\mu}(B_j, \xi)} \ge 2 \sqrt{\sprs\log(\sprs/\eps)} \;\middle|\; \xi,\mathcal{E}  } \le 2 \Paren{\frac{\eps}{\sprs}}^{2}.
        \]
        Hence assuming $\eps \leq 1/16$, we can upper bound the clipping error by
       \begin{align*}
           \abs{\bEE{\operatorname{Clip}_{\Delta}\Paren{\bar{X}_j}\mid \xi} - \bar{\mu}(B_j, \xi)} & \le \sprs \cdot \bPr{\bar{X}_j \notin [-\Delta, \Delta]} 
           \le \frac{2\eps^2}{s}\le \frac{\eps}{8}\,,
       \end{align*}
       concluding the proof.
    \end{proof}     
     With this claim in hand, we can analyze the detection step as follows. We have, after the above transformation and conditioned on $\xi$, each user obtains $(M\idx{1}, M\idx{2}, \ldots, M\idx{T})$ from a product distribution over $\bool^T$. By the ``simulate-and-infer'' trick, the mean vector of the product distribution can be learned to $\lp[\infty]$ distance $\frac{\dst}{20\sqrt{\sprs \log(s/\eps)}}$ with 
     \[
     \bigO{\frac{T\log T}{\min(T,\numbits)(\dst/\sqrt{\sprs \log(s/\eps)})^2}} 
     = \bigO{\frac{\dims\log(\dims/\sprs)\log(s/\eps)}{\min(\dims/\sprs,\numbits)\dst^2}}
     \]
     samples, allowing us to detect (with probability at least $99/100$) the at most $2$ biased coordinates. Of course, overall, we may only detect them when the choice of $\xi$ was good (so that the coordinates corresponding to the (at most two) biased blocks ended up indeed $\Omega(\dst/\sqrt{\sprs})$-biased); but since this happens with constant probability, one can pay a constant factor in the sample complexity and amplify this, to get a $99/100$ success probability overall. This concludes the proof. 
\end{proof}

Finally, we prove the matching lower bound (up to logarithmic factors).
\begin{lemma}\label{lem-mean-est-lower}
  For any $\sprs \geq 1$, the \emph{interactive} sample complexity of mean estimation of $\sprs$-block sparse product distributions over $\bool^\dims$ under $\numbits$-bit communication constraints is
  \smash{$
      \bigOmega{\frac{\sprs^2 + \dims}{\dst^2\numbits} + \frac{\sprs}{\dst^2}}
  $}.
\end{lemma}
\begin{proof}
The $\bigOmega{\frac{\sprs}{\dst^2}}$ term is simply the (unconstrained) ``oracle bound,'' as $\Omega(\frac{\sprs}{\dst^2})$ samples are required even without communication constraints and knowing which block of coordinates corresponds to the support of the mean vector. 

The $\Omega(\frac{\dims}{\dst^2\numbits})$ term follows from the case of $1$-sparse estimation~\citep[Theorem~2]{Shamir:14} (since, again, any $1$-sparse product distribution is $\sprs$-block-sparse for any $\sprs\geq 1$). Finally, the last term of the lower bound, $\Omega(\frac{\sprs^2}{\dst^2\numbits})$, follows from the lower bound on mean estimation under communication constraints (see, \eg \citet[Theorem~3]{AcharyaCST21}) in the \emph{nonsparse} case with $\dims = \sprs$, since even knowing the location of the block we still have a mean estimation task under information constraints, with dimensionality $\sprs$.
\end{proof}

\bibliographystyle{plainnat}
\bibliography{references-aos}

\appendix

\section{Results for the local privacy setting}
\label{sec:ldp}

In this section, we discuss how the results can be extended to the local privacy setting (LDP). In particular, we will focus on estimating the mean of sparse product distributions over $\bool^\dims$. The results on the block-sparse case will follow similarly. Under LDP constraints, each observation $X_\ui$ must be privatized using an $\priv$-LDP channel to get $Y_\ui$, which the estimate is based on.
\begin{definition}
For $\priv>0$, a channel $W\colon \cX \to \cY$ is said to be \emph{$\priv$-LDP} if, for all $x, x' \in \cX$ and $y \in \cY$,
\[
    \frac{W(y \mid x)}{W(y \mid x')} \le e^\priv. 
\]
\end{definition}
We focus on the high privacy regime, \ie when $\rho = O(1)$, and state the results below. Note that, in this regime, $(e^\priv - 1)^2 = O(\priv^2)$.

\begin{theorem}
  \label{theo:bernoullimean:sparse:LDP}
  For any $\sprs \geq 4\log\dims$, any $\rho$-LDP noninteractive protocol for mean estimation of $\sprs$-sparse product distributions over $\bool^\dims$ must have sample complexity 
  $
      \bigOmega{\frac{\sprs\dims}{\dst^2\priv^2}\log\frac{e\dims}{\sprs}}
  $.
\end{theorem}
Combined with previously known results for
sparse mean estimation, this lower bound immediately implies the following:
\begin{corollary}
  \label{coro:bernoullimean:sparse:LDP}
  For any $\sprs \geq 4\log\dims$, the \emph{noninteractive} sample complexity of mean estimation of $\sprs$-sparse product distributions over $\bool^\dims$ under $\priv$-LDP constraints is
  \smash{$
      \bigTheta{\frac{\sprs\dims}{\dst^2\priv^2}\log\frac{e\dims}{\sprs}}
  $}, while the \emph{interactive} sample complexity is
  \smash{$
      \bigTheta{\frac{\sprs\dims}{\dst^2\priv^2}}
  $}.
\end{corollary}

Of these bounds, the interactive upper and lower bounds are shown in \citet{AcharyaCST21} and \citet{DR:19}.
The noninteractive upper bound was established in~\citet{DJW:17}. The proof of \cref{theo:bernoullimean:sparse:LDP}, the noninteractive lower bound, follows similar steps as the proof of \cref{theo:bernoullimean:sparse}\footnote{For the case of $s = 1$, a lower bound of $\Omega\Paren{\frac{d}{\dst^2\priv^2}}$ is shown in \cite{ullman2018private}.}. We now discuss how to modify the argument for estimation under LDP constraints.

We first follow the same steps as in the proof of \cref{theo:bernoullimean:sparse} until \cref{eq:sum:information}, which we write below.
\begin{align*}
	\mutualinfo{Z}{Y_t}\le \sum_{y \in 
		\cY}  \Big( \frac{\sprs\gamma^{2}}{2\dims} \sum_{i \in [\dims]} \frac{\bE{X}{W(y\mid X)X_i
			}^2}{\bE{X}{W(y\mid X)}}  + \sum_{ r= 2}^{\dims} \Paren{\frac{\sprs\gamma^{2}}{2\dims}}^r \sum_{\substack{B 
			\subseteq [\dims]\\ |B|  = r}}  \frac{\bE{X}{W(y\mid X)\prod_{i \in B}X_i
			}^2}{\bE{X}{W(y\mid X)}} \Big) .
\end{align*}
To bound each term, we need the lemma below, proved in \citet{AcharyaCST21}, which follows from direct application of Bessel's inequality.
\begin{lemma} \label{lem:bessel:variance}
    Let $\phi_i\colon \cX \to \R$, for $i\leq 1$, be a family of functions. If the functions satisfy, for all $i, j$,
    \[
        \bE{X}{\phi_i(X) \phi_j(X)} = \indic{i = j},
    \]
    then, for any $\priv$-LDP channel $W$, we have
    \[
        \sum_i \bE{X}{\phi_i(X)W(y \mid X)}^2 \le \var_{X}\left[W(y \mid X) \right].
    \]
\end{lemma}
Note that for the first term,
\begin{align*}
	\sum_{y \in  \cY} \frac{\sprs}{2\dims} \gamma^2 \sum_{i \in [\dims]} 
	\frac{\bE{X\sim 
			\uniform}{W(y\mid X)X_i
		}^2}{\bE{X}{W(y\mid X)}}  & \le \frac{\sprs}{2\dims}  \gamma^2 \sum_{y \in 
	\cY}
	\frac{\var_{X}\left[W(y \mid X) \right]}{\bE{X}{W(y\mid X)}} \\
	& \le\frac{\sprs}{2\dims}  \gamma^2 \sum_{y \in 
	\cY}
	\frac{(e^\priv - 1)^2\bE{X}{W(y\mid X)}^2}{\bE{X}{W(y\mid X)}} \\
	& = \frac{\dst^2 \numbits}{2 \dims} (e^\priv - 1)^2.
\end{align*}
As in the proof of \cref{theo:bernoullimean:sparse}, we can use Lemma~\ref{lem:bessel:variance} to bound the $j$th order term by $
 (e^\priv-1)^2	\frac{\numbits \dst^2}{2 \dims}  (\dst^2/2)^{j - 1}.$ And thus, summing over all terms, we get
\[
		\mutualinfo{Z}{Y_t} \leq (e^\priv - 1)^2	\frac{\numbits \dst^2}{2 \dims}\sum_{j=1}^\infty (\dst/2)^{2(j-1)} \le  (e^\priv - 1)^2	\frac{\numbits \dst^2}{\dims}.
\]
Since $\mutualinfo{Z}{Y^\ns} \le \sum_{t = 1}^\ns \mutualinfo{Z}{Y_t} $, we conclude the proof using \cref{lem:sum_cond_mi}, thus establishing~\cref{theo:bernoullimean:sparse:LDP}.
\section{One-sparse noninteractive lower bound}
    \label{app:bernoullimean:1sparse}
In this section, we prove the following result for 1-sparse estimation under communication constraints.
\begin{theorem}
  \label{theo:bernoullimean:1sparse}
  Any $\numbits$-bit noninteractive protocol for mean estimation of 1-sparse product distributions over $\bool^\dims$ must have sample complexity 
  $
      \bigOmega{\frac{\dims\log \dims}{\dst^2\numbits}}
  $.
\end{theorem}
\begin{proof}
    Consider the following family of distributions. For $i \in [\dims]$, $\p_i$ is a product distribution over $\bool^\dims$ with mean $\theta_j = 2 \eps\indic{i=j}$ for $1\leq j\leq \dims$. Consider the generative process where we first sample $J$ uniformly from $[\dims]$ and then each user observes one sample from $\p_J$ and follows the protocol, thus obtaining a tuple $Y^\ns$ of messages.
    
    \noindent By Fano's inequality, we have that for any $1$-sparse estimation protocol the following must hold:
    \[
        \mutualinfo{J}{Y^\ns} = \Omega\Paren{\log \dims}.
    \]
    It remains to provide an upper bound on $\mutualinfo{J}{Y^\ns}$. Since $(Y_1, Y_2, \ldots, Y_\ns)$ are independent conditioned on $J$, we have
    \[
        \mutualinfo{J}{Y^\ns} \le \sum_{t = 1}^\ns \mutualinfo{J}{Y_t}\,,
    \]
    and therefore it suffices to show that
    $
        \mutualinfo{J}{Y_t} = O\Paren{\frac{\eps^2 \numbits}{\dims}}
    $
    for every $t \in [\ns]$. 
        Using the same notation as in the proof of Lemma~\ref{lemma:lb:sparse:noninteractive} we have
    \[ 
        \mutualinfo{J}{Y_t} \le \bE{J}{\kldiv{W_t^{\p_J}}{W_t^{\uniform}}} \le \bE{J}{\chisquare{W_t^{\p_J}}{W_t^{\uniform}}}.
    \]
    Now, we can expand 
    \begin{align*}
        \bE{J}{\chisquare{W_t^{\p_J}}{W_t^{\uniform}}}&  = 	\bE{J}{\sum_{y \in 
			\cY}\frac{\Paren{\sum_x W(y\mid x)(\p_{J}(x) -
				\uniform(x))}^2}{\sum_x 
			W(y\mid x) \uniform(x)}} \\
			&  = 	\bE{J}{\sum_{y \in 
			\cY}\frac{\bE{\uniform}{ W(y\mid X) 2 \dst X\idx{J}}^2}{\bE{\uniform}{
			W(y\mid X)}}} \\
			& = \frac{4\dst^2}{\dims} \sum_{y \in 
			\cY} \sum_{j \in [\dims]}\frac{\bE{\uniform}{ W(y\mid X) X\idx{j}}^2}{\bE{\uniform}{
			W(y\mid X)}} \\
			& = \frac{4\dst^2}\dims \sum_{y \in 
			\cY} \frac{\norm{\bE{\uniform}{ W(y\mid X) X}}_2^2}{\bE{\uniform}{
			W(y\mid X)}}\,.
    \end{align*}
    Note that the uniformly random vector $X$ is 1-subgaussian. Hence using Lemma~\ref{l:basic_mc}, we get $\mutualinfo{J}{Y_t} \le \frac{8( \ln 2) \eps^2 \numbits }{\dims}$, which lets us conclude the proof as we get that $\ns$ must satisfy
    $
        \Omega\Paren{\log \dims} = \mutualinfo{J}{Y^\ns} \le \sum_{t = 1}^\ns \mutualinfo{J}{Y_t} \leq \ns\cdot \frac{8( \ln 2) \eps^2 \numbits }{\dims}\,.
    $
\end{proof}

\section{Adaptive sensing from \texorpdfstring{$\nlm$}{low}-dimensional measurements} \label{sec:adaptive_sensing}
In this section, we prove \cref{theo:compressive}. The theorem states that there is an algorithm which estimates a sparse signal up to $\lp[2]$ accuracy $\eps$ with $\nlm \cdot \ns = O\Paren{\frac{\sprs\dims}{\eps^2} + \frac{\sprs}{\dst^2}\log \frac{e\dims}{\sprs} }$ noisy linear measurements, which is optimal as shown by the adaptive sensing lower bound from \citet{AriasCD12} and the sparse mean estimation lower bound from the centralized case (see, \eg \citet[Section~19]{Wu20}). Moreover, as shown in \cite{RaskuttiWY11}, $\nlm \cdot \ns = \Omega \Paren{\frac{\sprs\dims}{\eps^2} \log \frac{e\dims}{\sprs}}$ measurement are required for a noninteractive protocol. All together, this demonstrates a separation between noninteractive compressed sensing and adaptive sensing. \smallskip

\begin{proofof}{\cref{theo:compressive}.} 
    We establish the result by a reduction to estimating the mean of a sparse product distribution over $\bool^\dims$, which we have considered in previous sections.
    
    Let $e_i$ be the $i$th standard base vector in $\R^\dims$. Consider the family of selection matrices containing, for every $S \subseteq [\dims]$ of size $|S| = \nlm$, the matrix  $A_S \eqdef [e_i]_{i \in S}$.
    Then by \cref{task:compressivesensing}, for any $S \subseteq [\dims]$, $Y \sim \gaussian{x_{S}}{I_{\nlm}}$, where $x_{S}\in\R^{\nlm}$ denotes the subvector of $x$ restricted to coordinates indexed by $S$. Let $Y' = (\sign(Y\idx{i}))_{i \in S}$. Then $Y'$ has a product distribution such that, for every $i \in S$, $Y'_i\in\bool$ has mean 
    \[
        \mu\idx{i} \eqdef \bEE{Y\idx{i}'} = 2\bPr{Y\idx{i}> 0} - 1 = \operatorname{Erf}\Paren{\frac{x\idx{i}}{\sqrt{2}}},
    \]
    where $\operatorname{Erf}$ is the Gaussian error function. For $x \in [-1, +1]^\dims$, $\mu \in [-\operatorname{Erf}(1/\sqrt{2}), \operatorname{Erf}(1/\sqrt{2})] \subset [-1, +1]^\dims$. We will rely on the following lemma from \citet{AcharyaCST21}, which states that a good estimate for $\mu$ is also a good estimate for $x$.
    \begin{lemma}[{\citet[Lemma~7]{AcharyaCST21}}]
        For  $\widehat{\mu}\in[-\eta,\eta]^\dims$, define $\widehat{x}\in[-1,1]^\dims$ by
    $
        \widehat{x}\idx{i} \eqdef \sqrt{2}\operatorname{Erf}^{-1}(\widehat{\mu}\idx{i}),
    $ for all $i\in[\dims]$. 
    Then
    $
        \norm{\widehat{x}-x}_2 \leq \sqrt{\frac{e\pi}{2}}\cdot \norm{\mu -\widehat{\mu}}_2\,.
    $
    \end{lemma}
    It only remains to establish an upper bound on estimating the mean of a product distribution over $\bool^\dims$ by observing a subset of $\nlm$ coordinates from each sample (in particular, this is a more restricted constraint than $\numbits$-bit communication, where the message is not restricted to consist of bits of the original sample). Nonetheless, in the protocol in \cite{AcharyaCST21} which achieves \cref{lemma:ub:sparse:interactive}, each user does actually send $\numbits$ coordinates of the observed sample, meaning that it can be directly applied here by setting $\numbits = \nlm$. Plugging $\nlm$ for $\numbits$ in \cref{lemma:ub:sparse:interactive}, we get the desired bound.
\end{proofof}

\section{Missing proofs in \cref{sec:sparse_mean,sec-block-sparse}} \label{sec:missing-proof}
We now provide the proofs of the two inequalities used in~\cref{sec:sparse_mean,sec-block-sparse}, respectively.
\begin{proofof}{\cref{eq:sum:information}}
We can rewrite and bound the mutual information as
\[
		\mutualinfo{Z}{Y_t}  = \bE{Z}{\kldiv{W_t^{\p_{Z}}}{W_t^{\uniform}}}
\le 
\bE{Z}{\chisquare{W_t^{\p_{Z}}}{W_t^{\uniform}}}\,. 
\]
We drop the subscript $t$ from $W_t$ when it is clear from context.
Using the definition of chi-square divergence and~\cref{eq:def:induced:distribution},
for $X,X'$ generated \iid from $\uniform$, we have 
\begin{align*}
	&\bE{Z}{\chisquare{W_t^{\p_{Z}}}{W_t^{\uniform}}}\\
	&\qquad =	\bE{Z}{\sum_{y \in 
			\cY}\frac{\Paren{\sum_x W(y\mid x)(\p_{Z }(x) - 
				\uniform(x))}^2}{\sum_x 
			W(y\mid x) \uniform(x)}} \\
	&\qquad = \sum_{y \in 
		\cY}	\bE{Z}{\frac{\bE{\uniform}{W(y\mid X)\Paren{\prod_{i=1}^{\dims} (1 
					+ 
					\gamma Z\idx{i} X\idx{i}) - 1} }^2}{\bE{\uniform}{W(y\mid X)}}} \\
	&\qquad = \sum_{y \in 
		\cY}	\bE{Z}{\frac{\bE{X, X' \sim \uniform}{W(y\mid X)W(y\mid X')\Paren{\prod_{i=1}^{\dims} (1 + 
					\gamma Z\idx{i} X\idx{i}) - 1} \Paren{\prod_{i=1}^{\dims} (1 + 
					\gamma Z\idx{i} X\idx{i}') - 1}  }}{\bE{\uniform}{W(y\mid X)}}},
\end{align*}
where we recall that $\gamma = \dst/\sqrt{\sprs}$.
Note that since $\bE{Z}{Z\idx{i}} = 0$ and $\bE{Z}{Z\idx{i}^2} = \frac{\sprs}{2\dims}=\sigma^2$ for all $i \in [\dims]$ 
	and the $Z\idx{i}$'s are independent, we further obtain that
\begin{align*}
	\bE{Z}{\mright. & \mleft.\Paren{\prod_{i=1}^{\dims} (1 + 
			\gamma X\idx{i} Z\idx{i}) - 1} \Paren{\prod_{i=1}^{\dims} (1 + 
			\gamma X\idx{i}' Z\idx{i}) - 1} } \\
	& = \bE{Z}{\prod_{i=1}^{\dims}(1 + 
		\gamma Z\idx{i} X\idx{i}) (1 + 
		\gamma Z\idx{i} X\idx{i}')} 
	- 
	2\bE{Z}{\prod_{i=1}^{\dims} (1 + 
		\gamma Z\idx{i} X\idx{i}) } 
		+ 1 \\
	& = \prod_{i=1}^{\dims} (1 + \sigma^2 
	\gamma^2 X\idx{i}X\idx{i}') - 1\,.
\end{align*}
Plugging this into the previous expression, we get
\begin{align*}
	 \bE{Z}{\chisquare{W_t^{\p_{Z}}}{W_t^{\uniform}}} 
	& = \sum_{y \in \cY}	\frac{\bE{X, X' \sim \uniform}{W(y\mid X)W(y\mid X')\Paren{\prod_{i=1}^{\dims} (1 + 
				 \sigma^2 \gamma^2 X\idx{i} 
				X\idx{i}') - 1}  }}{\bE{\uniform}{W(y\mid X)}}  \\
	& = \sum_{y \in \cY} 	\frac{\bE{X, X' \sim \uniform}{W(y\mid X)W(y\mid X')\Paren{\sum_{ r= 1}^{\dims} \sum_{B 
					\subseteq [\dims], |B|  = r} \sigma^{2r}\gamma^{2r} \prod_{i \in B}X\idx{i} 
					X\idx{i}'} 
	}}{\bE{\uniform}{W(y\mid X)}}  \\
	& = \sum_{y \in \cY}  \Big(\sigma^2\gamma^2\sum_{i \in [\dims]} \frac{\bE{X,X' \sim \uniform}{W(y\mid X)W(y\mid X')X\idx{i}X\idx{i}' }}{\bE{\uniform}{W(y\mid X)}}  \notag\\
	&\qquad\qquad+ \sum_{ r= 2}^{\dims} \sum_{\substack{B \subseteq [\dims]\\ |B|  = r}} \sigma^{2r}\gamma^{2r}
		\frac{\bE{X, X'\sim \uniform}{W(y\mid X)W(y\mid X')\prod_{i \in B}X\idx{i}X\idx{i}'
		}}{\bE{\uniform}{W(y\mid X)}} \Big)  \\
	& = \sum_{y \in \cY}  \!\!\Big( \sigma^2 \gamma^2 \sum_{i \in [\dims]} \frac{\bE{\uniform}{W(y\mid X)X\idx{i}
			}^2}{\bE{\uniform}{W(y\mid X)}}  + \sum_{ r= 2}^{\dims} \sigma^{2r}\gamma^{2r} \!\!\!\sum_{\substack{B 
			\subseteq [\dims]\\ |B|  = r}}  \frac{\bE{\uniform}{W(y\mid X)\prod_{i \in B}X\idx{i}
			}^2}{\bE{\uniform}{W(y\mid X)}} \Big)\,,
\end{align*}
which is the inequality we wanted to establish.
\end{proofof}

\begin{proofof}{\cref{eq:bigstuff}}
	Let $\uniform$ denote the uniform distribution over $\bool^\dims$, and $\gamma \eqdef \dst/\sqrt{\sprs}$. For 
	all $j \in [b]$, we have
	\begin{align}
		&\condmutualinfo{Z}{Y_t}{J=j} \notag\\
		\;\;& \le \bE{Z}{\kldiv{W^{\p_{Z,j}}}{W^\uniform}} \notag\\
		\;\;&\le 
		\bE{Z}{\chisquare{W^{\p_{Z,j}}}{W^\uniform}} \notag\\
		\;\;& =	\bE{Z}{\sum_{y \in 
				\cY}\frac{\Paren{\sum_x W(y\mid X)(\p_{Z,j}(x) - 
				\uniform(x))}^2}{\sum_x 
				W(y\mid X) \uniform(x)}} \notag\\
		\;\;& = \sum_{y \in 
			\cY}	\bE{Z}{\frac{\bE{X}{W(y\mid X)\Paren{\prod_{i \in B_j} (1 + 
						\gamma Z\idx{i} X\idx{i}) - 1} }^2}{\bE{X}{W(y\mid X)}}} \notag\\
		\;\;& = \sum_{y \in 
			\cY}	\bE{Z}{\frac{\bE{X,X'}{W(y\mid X)W(y\mid X')\Paren{\prod_{i \in 
							B_j} (1 + \gamma Z\idx{i} X\idx{i}) - 1} \Paren{\prod_{i \in 
							B_j} (1 + 
						\gamma Z\idx{i} X\idx{i}') - 1}  }}{\bE{X}{W(y\mid X)}}} \notag\\
		\;\;& = \sum_{y \in 
			\cY}	\frac{\bE{X, X'}{W(y\mid X)W(y\mid X')\bE{Z}{\Paren{\prod_{i \in 
							B_j} (1 + 
						\gamma Z\idx{i} X\idx{i}) - 1} \Paren{\prod_{i \in 
							B_j} (1 + 
						\gamma Z\idx{i} X\idx{i}') - 1} } }}{\bE{X}{W(y\mid X)}},
		\label{eqn:expect_product}
	\end{align}
	where $X,X'\sim\uniform$ are independent.
	Note that since $\bE{Z}{Z\idx{i}} = 0$ and $\bE{Z}{Z\idx{i}^2} = 1$ for all $i \in [\dims]$ 
	and the $Z\idx{i}$'s are independent, we have 
	\begin{align*}
		&\bE{Z}{\Bigg(\prod_{i \in 
					B_j} (1 + 
				\gamma Z\idx{i} X\idx{i} ) - 1 \Bigg) \Bigg(\prod_{i \in 
					B_j} (1 + 
				\gamma Z\idx{i} X\idx{i}') - 1 \Bigg) } \\
		&\quad\quad\quad =\prod_{i \in 
			B_j} \bE{Z}{(1 + 
				\gamma Z\idx{i} X\idx{i}) (1 + 
			\gamma Z\idx{i} X\idx{i}')} - 1
		 = \prod_{i \in 
			B_j} (1 + 
		\gamma^2X\idx{i}X\idx{i}') - 1\,.
	\end{align*}
	Plugging this into~\cref{eqn:expect_product}, we obtain 
	\begin{align*}
		\condmutualinfo{Z}{Y_t}{J=j} & \le \sum_{y \in 
			\cY}	\frac{\bE{X, X'}{W(y\mid X)W(y\mid X')\Paren{\prod_{i \in 
						B_j} (1 + 
					\gamma^2 X\idx{i} 
					X\idx{i}') - 1}  }}{\bE{X}{W(y\mid X)}} \\
		& = \sum_{y \in 
			\cY} 	\frac{\bE{X, X'}{W(y\mid X)W(y\mid X')\Paren{\sum_{ r =1}^{\sprs} \sum_{B 
						\subseteq B_j, |B| = r} \gamma^{2r} \prod_{i \in B}X\idx{i} X\idx{i}'} 
		}}{\bE{X}{W(y\mid X)}} \\
		& = \sum_{y \in 
			\cY}  \left(\gamma^2\sum_{i \in B_j} \frac{\bE{X, X'}{W(y\mid X)W(y\mid X')X\idx{i}X\idx{i}'
			}}{\bE{X}{W(y\mid X)}} \right. \\
				& \;\;\;\;\;\;\; \;\;\;\;\;\;\;\left. + \sum_{ r =2}^{\sprs} \sum_{\substack{B\subseteq B_j\\ |B| = r}}  \gamma^{2r}
			\frac{\bE{X, X'}{W(y\mid X)W(y\mid X')\prod_{i \in B}X\idx{i}X\idx{i}'
			}}{\bE{X}{W(y\mid X)}} \right) \\
		& = \sum_{y \in 
			\cY}  \Big(\gamma^{2} \sum_{i \in B_j} \frac{\bE{X}{W(y\mid X)X\idx{i}
				}^2}{\bE{X}{W(y\mid X)}} + \sum_{ r =2}^{\sprs} \gamma^{2r}  \sum_{\substack{B\subseteq B_j\\ |B| = r}}   \frac{\bE{X}{W(y\mid X)\prod_{i \in B}X\idx{i}
				}^2}{\bE{X}{W(y\mid X)}} \Big).
	\end{align*}
	Summing over all the blocks, we get
	\begin{align*}	
		\sum_{j\in [b]} \condmutualinfo{Z}{Y_t}{J=j} & \le \sum_{j \in [b]} 
		\sum_{y 
			\in 
			\cY}  \Big({\gamma^{2} \sum_{i \in B_j} \frac{\bE{X}{W(y\mid X)X\idx{i}
				}^2}{\bE{X}{W(y\mid X)}}  + \sum_{ r =2}^{\sprs} \gamma^{2r} \!\!\sum_{\substack{B\subseteq B_j\\ |B| = r}}\!\!  \frac{\bE{X}{W(y\mid X)\prod_{i \in B}X\idx{i}
				}^2}{\bE{X}{W(y\mid X)}} }\Big).
	\end{align*}
	which is what we set out to prove.
\end{proofof}

\end{document}